\documentclass[twocolumn, prX]{revtex4}
\usepackage{graphicx}
\usepackage{amssymb,amsmath,amsfonts,amsthm}
\usepackage{amssymb,mathtools}
\usepackage[sans]{dsfont}
\usepackage{epstopdf}
\usepackage{subfig}
\usepackage{epsfig}
\newcommand{\beq}{\begin{equation}}
\newcommand{\eeq}{\end{equation}}

\newcommand{\f}{\begin{equation}}
\newcommand{\ff}{\end{equation}}

\newtheorem{theorem}{Theorem}

\newtheorem{definition}{Definition}

\begin{document}

\title{Infinite Degeneracy of States in Quantum Gravity}
\author{Jonathan Hackett\thanks{Email address:
jhackett@perimeterinstitute.ca} and Yidun Wan\thanks{Email address:
ywan@perimeterinstitute.ca}}

\address{Perimeter Institute for Theoretical Physics,\\
31 Caroline st. N., Waterloo, Ontario N2L 2Y5, Canada, and \\
Department of Physics, University of Waterloo,\\
Waterloo, Ontario N2J 2W9, Canada\\}
\date{November 13, 2008}

\vfill
\begin{abstract}
The setting of Braided Ribbon Networks is used to present a general
result in spin-networks embedded in manifolds: the existence of an
infinite number of species of conserved quantities. Restricted to
three-valent networks the number of such conserved quantities in a
given network is shown to be determined by the number of nodes in the network. The
implication of these conserved quantities is discussed in the
context of Loop Quantum Gravity.
\end{abstract}

\maketitle
\section{Introduction}

Spin networks\cite{Penrose} have been a part of the search for
quantum gravity for over three decades. More recently they have
gained prominence in Loop Quantum Gravity as a basis for the
kinematical Hilbert space\cite{spin-foam, spin-foam2}. Though a
great deal of the understanding of spin-networks is derived directly
from graph theory, the spin-network states of quantum gravity
possess a richer structure owing to their embedding. Recently a
great deal of attention has been focused on understanding the
possibility that we may have to work with an even further enriched
structure of spin networks with width \cite{BilsonThompson:2005,
BilsonThompson:2006, Hackett2007, Wan2007, BilsonThompson:2008ex,
Wan:2008qs, LouNumber, Hackett:2008tt, LeeWan2007, He:2008is,
He:2008jc}. This work was motivated from the suggestion that in
studying loop quantum gravity with a positive cosmological constant
we may need to consider framed spin networks \cite{Major:1995yz,
Borissov:1995cn, Smolin:2002sz}.

One of the major revelations of the work on framed spin networks has
been the existence of conserved topological structures. This was
joined in \cite{Markopoulou:2008be} with the existence of conserved
quantities even in the absence of an embedding.

We shall present an overview of embedded three-valent framed spin
networks (which we will call three-valent Braided Ribbon Networks),
we will then describe the existence of conserved topological
structures in these networks and demonstrate the existence of
countably infinite species of these structures. We will then finally
lift the result to embedded three-valent spin networks without
framing. The results on the classification of the structures will be
general in the valence of the spin-network though the results in the
invariance of the number of conserved structures will only apply to
the three-valent case.

\section{Braided Ribbon Networks}
Braided ribbon networks (BRNs) generalize spin networks by extending
the edges into a higher dimensional structure and associating
non-isotopic embeddings of the same spin network to different basis
states in a quantum space of a theory of quantum gravity. Trivalent
BRNs are constructed through taking the union of trinions -
2-surfaces of the form of Fig. \ref{trin1}. By considering a
trinion to be the combination of three legs and a node (the union of
the three legs), we can allow the legs of the trinion to be twisted
about and to cross over parts of other trinions before being joined
together into a BRN. This gives us generalized structures with
features such as those in Fig. \ref{trin2}. We will use the word ribbon to refer to the extended edges of the BRN, and use the term edges to refer instead to the boundaries of the ribbons.  The BRN is considered
to evolve subject to the standard evolution algebra of trivalent
spin networks $\mathcal{A}_{evol}$ whose generators consist of the
$1-3$ move and the $2-2$ move (Fig. \ref{moves}).  Whether these correspond to evolution with respect to time, or some other concept is irrelevant to this work, and so will be left to the reader's taste.   In the present work we shall use what is referred to as the Smolin-Wan rule in \cite{Markopoulou:2008be} which means that a $2-2$ move cannot be performed when the shared ribbon (labeled $b$ in the exchange move of Fig. \ref{moves}) has additional content from the embedding (for example a knot, or a twist).  We shall also
consider the simpler case of embedded spin networks, an analogous
situation where we do not extend the graph edges, but still consider
non-isotopic embeddings of the same abstract spin network to be
distinct (these admit similar structures but without twisting).  In discussing embedded spin networks, we will use the word edge in the sense of graph theory - the meaning of the word edge should be clear from its context.

Additionally before we proceed we will clarify what we mean by a `knot' in an embedded spin network:
\begin{definition}
We will call a \textbf{local knot} any locally knotted arc that we can place a closed compact surface around with the network only intersecting the surface precisely twice at the single edge making up the locally knotted arc.
\end{definition}

\begin{figure}[!h]
  \begin{center}
  \subfloat[]{\label{trin1}\includegraphics[scale=0.2]{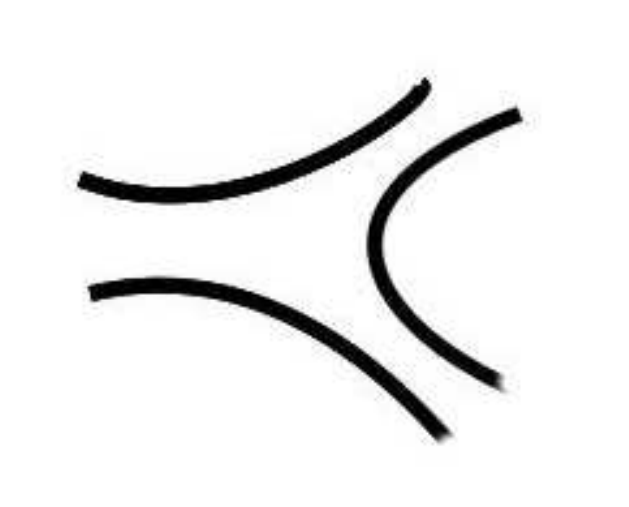}}
  \subfloat[]{\label{trin2}\includegraphics[scale=0.2]{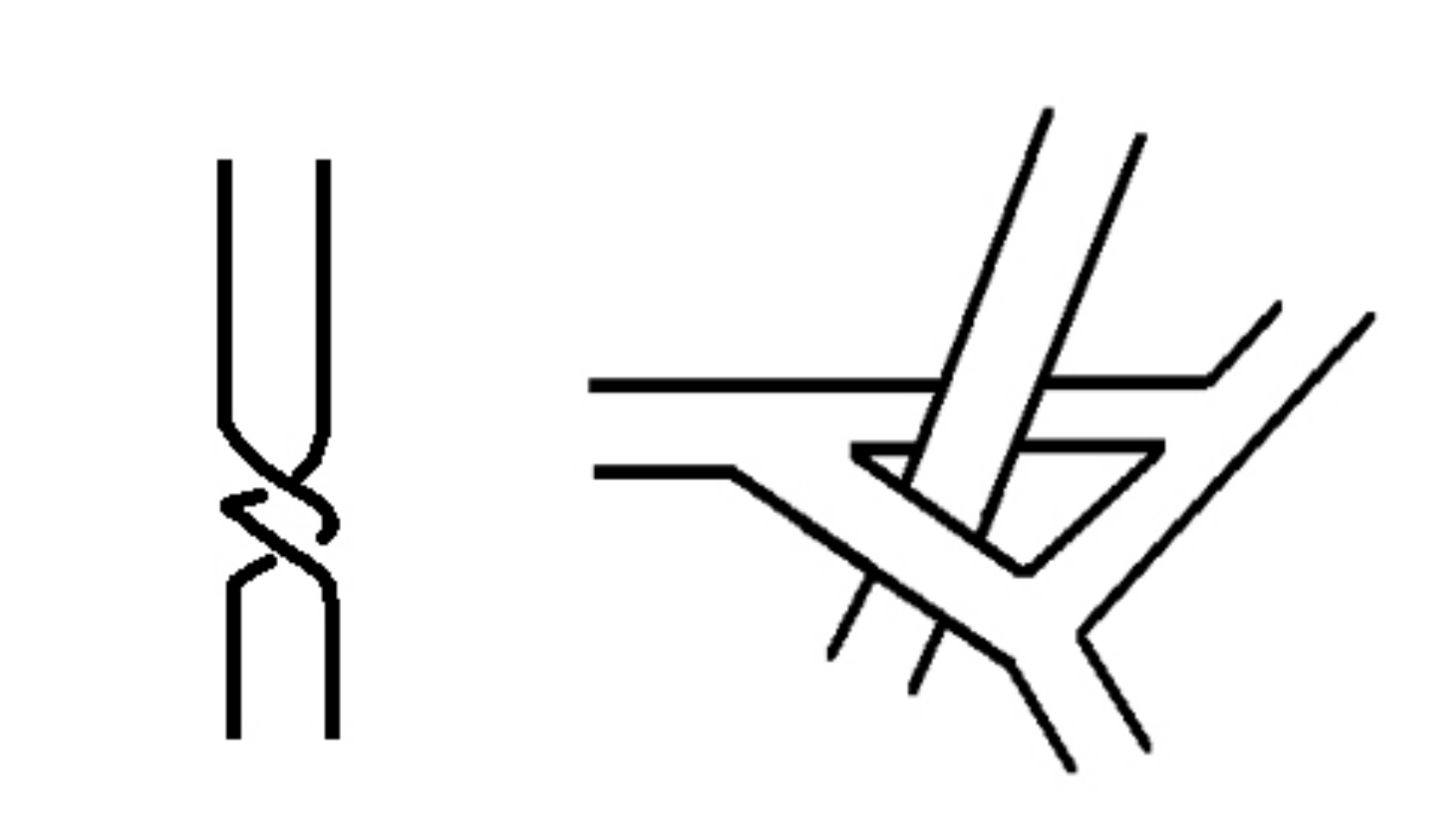}}
\end{center}
  \caption{Trinions, twists and braiding}
  \label{trinintro}
\end{figure}

\begin{figure}
  \begin{center}
    \includegraphics[scale=0.6]{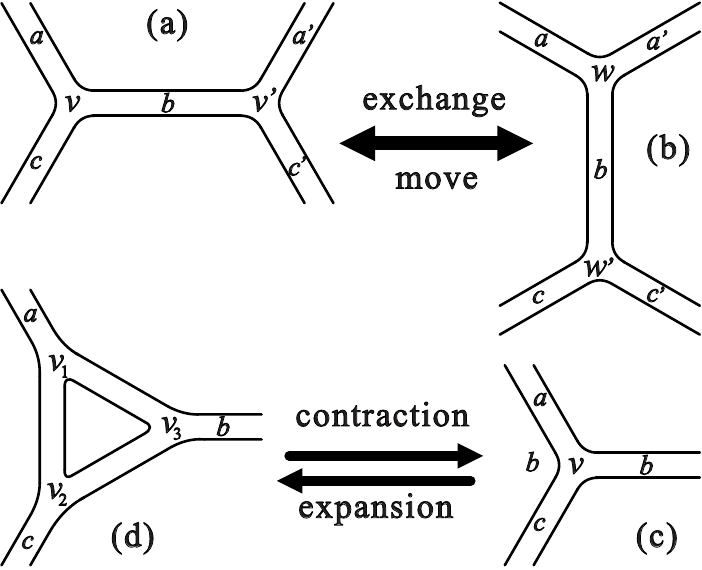}
  \end{center}
\caption{Generators of $\mathcal{A}_{evol}$}
\label{moves}
\end{figure}

In \cite{BilsonThompson:2006} the reduced link was introduced as a
tool for understanding the non-trivial topological content of a
braided network. The reduced link consists of removing the interior
of the 2-surface - leaving only the edges of the ribbons - and then removing all
un-linked un-knotted loops. An example of taking the reduced link is
shown in Fig. \ref{reducedlink} to demonstrate the process. It has
been shown that the reduced link is an invariant of the generators
of $\mathcal{A}_{evol}$ and therefore could be used to demonstrate
invariants of the BRN.

\begin{figure}
  \begin{center}
    \includegraphics[scale=0.3]{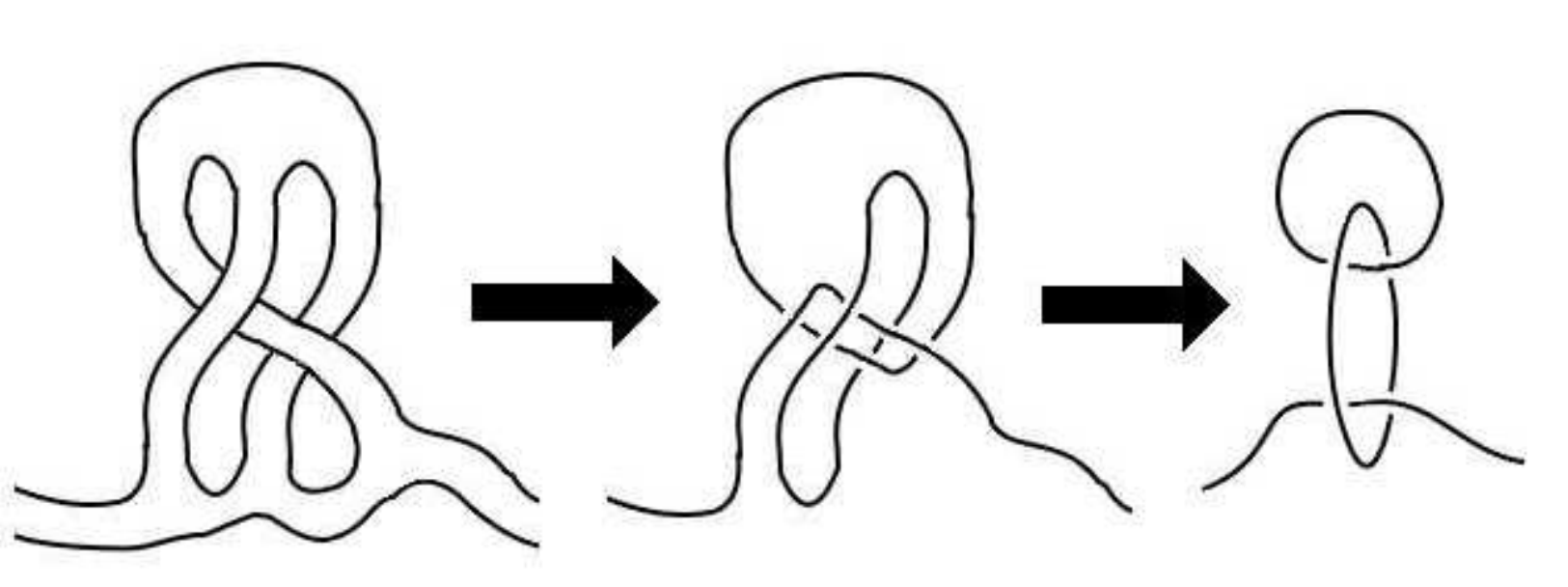}
  \end{center}
\caption{The Reduced Link}
\label{reducedlink}
\end{figure}

\section{Conserved Structures in Braided Ribbon Networks}

Given the existence of conserved structures in braided ribbon
networks we shall consider specifically the more local structures
(in the sense of micro-locality from \cite{Hackett2007}). This
specifically shall ignore structures where ribbons are knotted or
braided with ribbons that are distant under the standard distance
function of graph theory. We can consider the local structures in
order of reducing locality: the most local are those that involve
only a single ribbon (ultra local), then those that involve ribbons
sharing a node (1st degree local), then those that involve the ribbons
of two adjacent nodes (2nd degree local), and so on.  In particular,
we know from \cite{Markopoulou:2008be} that a structure such as that
in Fig. \ref{kedge} is conserved as we are unable to perform the $2-2$ move, and is
therefore an ultra local conserved structure.  The only exception to this is if the ribbon we are concerned with connects two 'halves' of the network (i.e. that the ribbon divides the network between two parts that are not connected to one another through anything other than that single ribbon - in the language of the next section, it corresponds to both $b$ and $b^\prime$ being tethers of isolated substructures).  In this situation we are able to remove the knotting and twisting by isotopy.  As this case is artificial in nature and uninteresting, we shall ignore it for our investigation.

\subsection{Isolating the conserved structures}

In \cite{Hackett2007} the concept of an isolated substructure was
introduced as a means of understanding the ability to translate
features through a braided ribbon graph. An isolated substructure is
a subset of a graph which connects to the rest of the graph only
through a single ribbon (called its tether) and is not part of any
larger topological features.  If a substructure can be evolved into an isolated substructure by a sequence of applications of elements of the evolution algebra, we will call them isolatable. Isolatable substructures are essentially propagating locally conserved quantities\cite{Hackett2007} able to move via the evolution algebra to any point edge-connected to its tether, and having conserved structure inside of it. We can also see via the form of the reduced link of an isolated substructure - and the invariance of the reduced link under the evolution moves - that an isolatable substructure corresponds to a 'piece' of the reduced link that is essentially cut and paste into the link of the edge its tether is on (see for example fig.\ref{reducedlink}).  That a general reduced link can be considered a direct product of these `pieces' means that a structure being isolatable does not require the evolution to acquire its meaning as a part of an invariant of the network, and that this meaning is invariant under interpretation of the meaning of the evolution.   We shall demonstrate that a specific class of ultra-local conserved quantities are isolatable.  To do this we shall first introduce several definitions.

\begin{figure}[h]
  \begin{center}
    \includegraphics[scale=0.6]{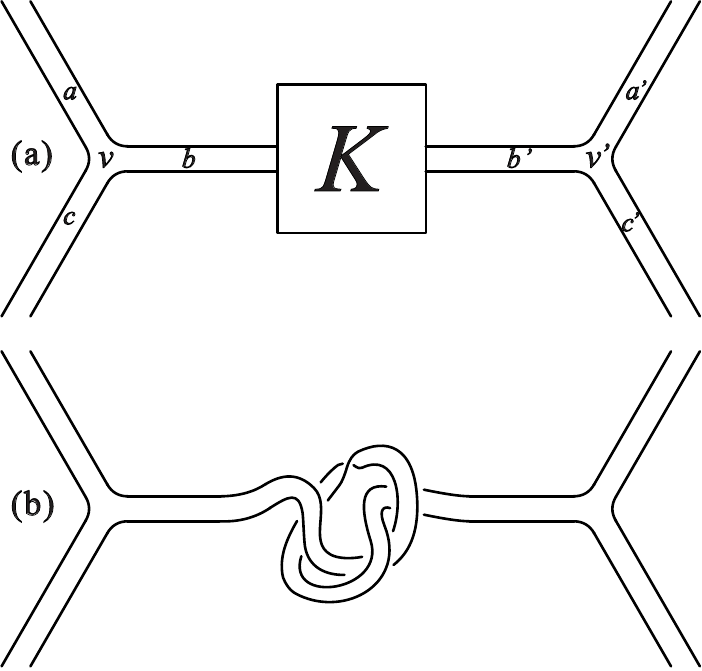}
  \end{center}
\caption{An example of an ultra-local structure} \label{kedge}
\end{figure}

\begin{definition}
Two edge segments are said to be \textbf{edge connected} if they are
connected in the space consisting of the edge of the network.
Equivalently two edge segments are said to be edge connected if they
are part of the same link in the reduced link of the network. The
path in the space of edges between the two edge segments is called
the \textbf{edge path}.
\end{definition}

\begin{definition}
Two ribbons or two nodes in a network have a \textbf{path} between
them if there exists a sequence of ribbons and nodes that can be
traversed between them. The sequence of ribbons and nodes taken is
called the path.
\end{definition}

\begin{definition}
A \textbf{free path} is a path which does not have any twists, knots
or links along it.  Specifically, each ribbon connecting the nodes of the path does not have any knotting or twisting on it, and there is no ribbon that crosses a ribbon in the path in such a way that cannot be undone by the Reidemeister moves applied to the ribbons.
\end{definition}

\begin{definition}
A \textbf{free edge path} is an edge path which does not have any
knots or links along it.  This corresponds similarly to requiring that each ribbon the edges of the path belong to does not have any knotting or twisting on it, and there is no ribbon that crosses one of these ribbons in such a way that cannot be undone by the Reidemeister moves applied to the ribbons.
\end{definition}

We shall now prove that a general class of ultra-local structures
(see Fig. \ref{kedge}) can be made isolated.

\begin{theorem}
An ultra-local structure in a BRN which possesses a free edge path
between one of the edges of each of the external ribbons can be isolated.
\end{theorem}
\begin{figure}
  \begin{center}
    \includegraphics[scale=0.4]{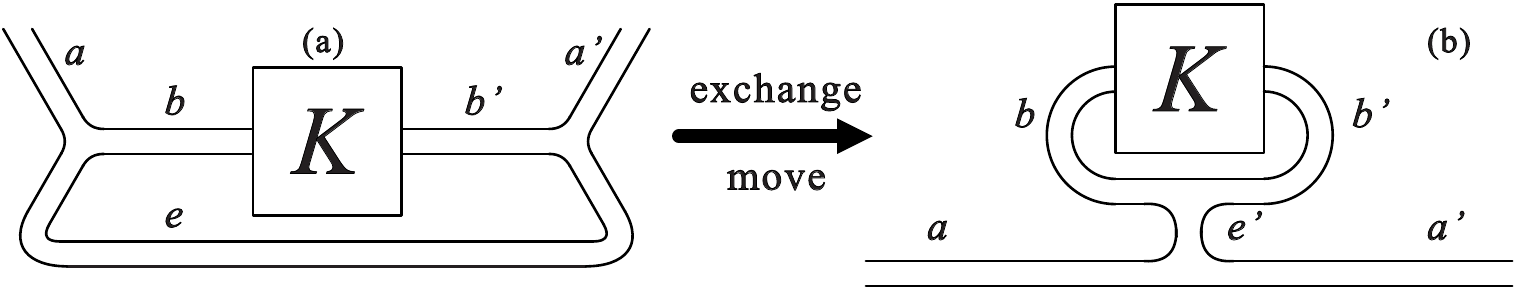}
  \end{center}
\caption{0 Node Case} \label{Theo1proofA}
\end{figure}

\begin{figure}
  \begin{center}
    \includegraphics[scale=0.4]{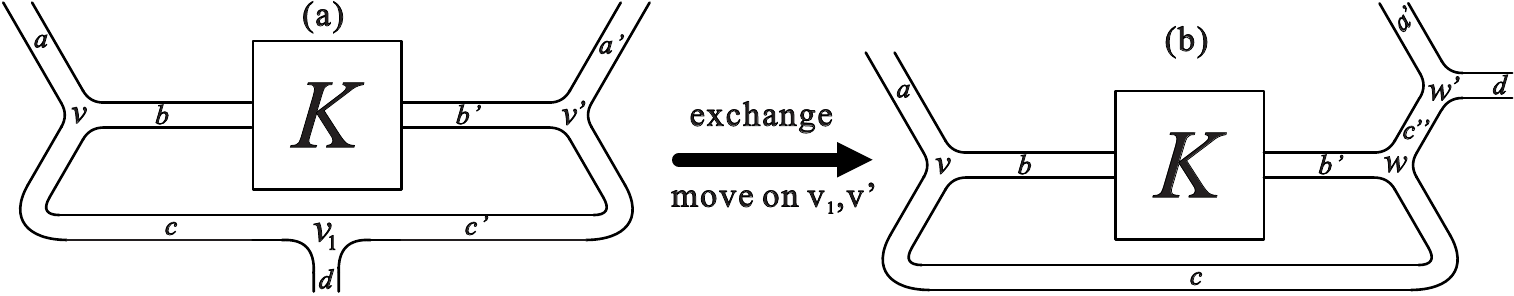}
  \end{center}
\caption{1 Node Case}
\label{Theo1proofB}
\end{figure}

\begin{proof}
We shall prove this using induction on the number of intervening
nodes. First we shall prove for one intervening nodes, then assume
true for $n-1$ nodes and prove true for $n$ nodes. Consider the
situation depicted in Fig. \ref{Theo1proofB} (where the apparent
orientation of the nodes is for simplicity, and is in fact general),
the application of the exchange move between nodes $v_1$ and
$v^\prime$ reduces the situation to that of Fig. \ref{Theo1proofA}
which can then be isolated by using the exchange move on the two
nodes involved. Now, we examine the situation in Fig.
\ref{Theo1proofC} to demonstrate that the $n$ node situation can be
reduced to $n-1$ nodes by applying the exchange move on $v^\prime$
and $v_n$. We can then use the assumption of truth on the $n-1$ case
to isolate the knot.
\end{proof}

\begin{figure}
  \begin{center}
    \includegraphics[scale=0.5]{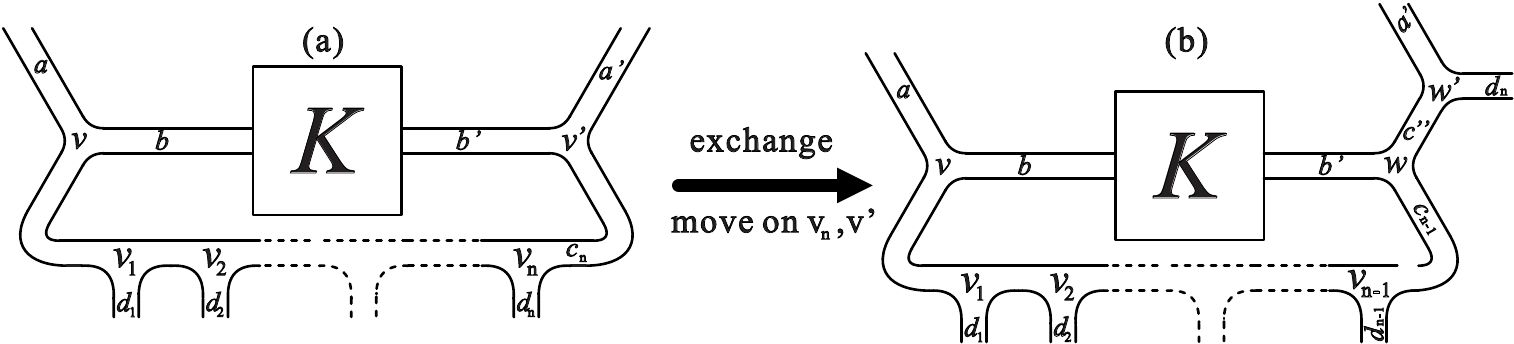}
  \end{center}
\caption{$n$ Node Case}
\label{Theo1proofC}
\end{figure}

This result lets us examine a peculiar situation: that where a knot
on a ribbon is \textit{not} a conserved quantity. Examining Fig.
\ref{except1} (where the unattached ribbons connect to a larger
network) we can see that it is possible in certain situations to
reduce the number of ultra local structures, in the sense that two
knots, e.g. the $K_1$ and $K_2$ in the figure, merge with each
other.  From this we see that we can always obtain a graph with only conserved local structures remaining.
\begin{figure}
  \begin{center}
  \includegraphics[scale=0.6]{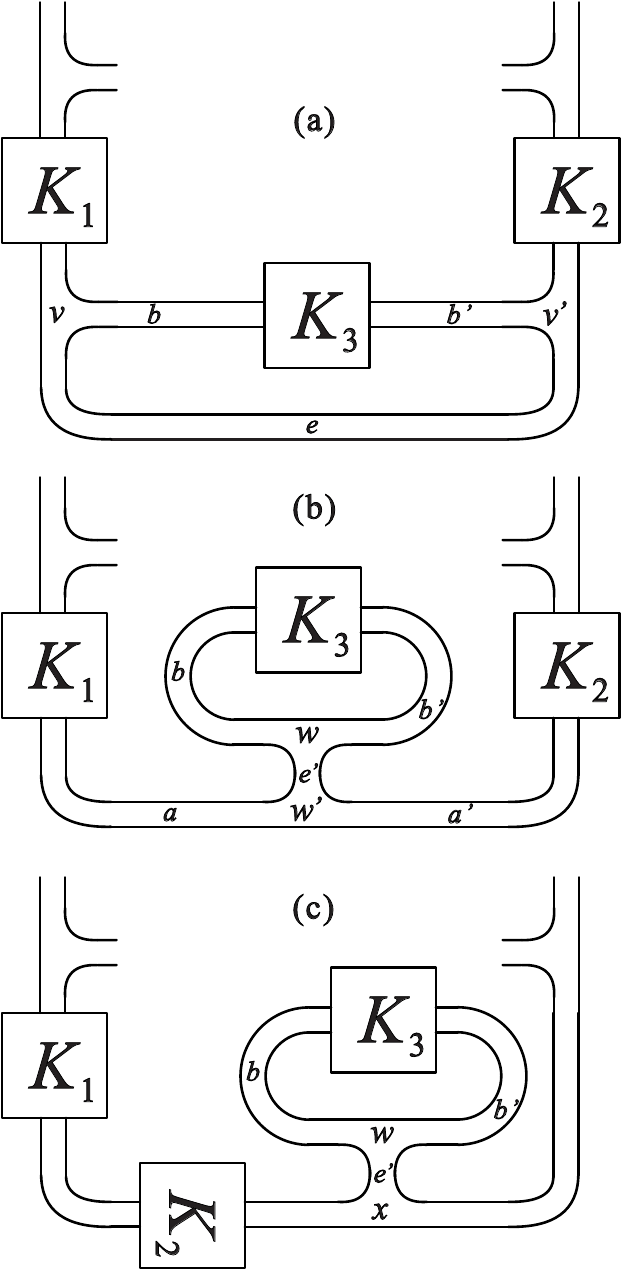}
  \end{center}
\caption{Reduction of the number of ultra local structures}
\label{except1}
\end{figure}

\begin{theorem}
A knot in an embedded graph which possesses a free path between its
two external edges can be isolated.
\end{theorem}
\begin{proof}
The proof of this follows inherently from the above proof and the
fact that in an embedded graph - instead of a BRN - one can rotate
an edge without introducing a twist.
\end{proof}

We can apply the above theorems to reduce less local structures to
more local situations. Consider for example the situation in Fig.
\ref{lEdge}, we can apply the results of the above theorems to
transform it to Fig. \ref{lEdge2to2} if ribbons $a$ and $b$, and
$a^\prime$ and $b^\prime$ are connected by a free edge path.

\begin{figure}[h]
  \begin{center}
  \subfloat[]{\label{lEdge}\includegraphics[scale=0.6]{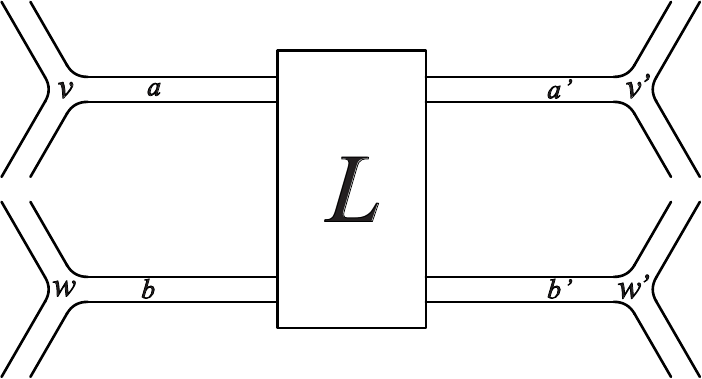}}\ \ \
  \ \
  \subfloat[]{\label{lEdge2to2}\includegraphics[scale=0.6]{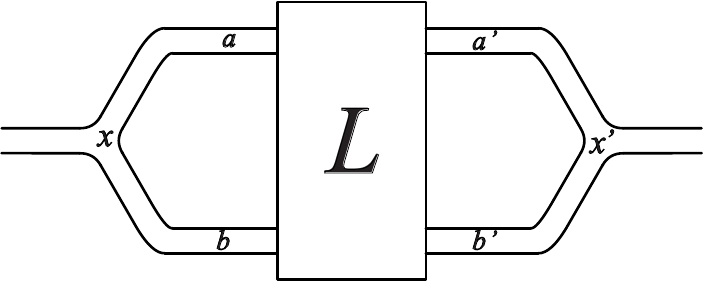}}
  \end{center}
\caption{Less Local knottings}
\label{lEdges}
\end{figure}

We can also see from these results that if there is no (edge) free path between the two edges of any of the knots in a network, there is no means to combine the knots onto single edges.  This leads us to conclude that if one isolates all isolatable knots on a graph the remaining knots are completely invariant.

\subsection{Immediate Results of Ultra-Local Structures}

Considering the idea of ultra local structures, we find that there
is only a single type of conserved structure. A structure formed by
the topological deformation of a single ribbon can only possess two
features: knots and twists. As twists can be passed through the
knotting of a ribbon by an isotopy, we can consider any ultra local
structure to be exactly characterized by a half integer
(corresponding to the number of rotations) and a knot or a connected
sum of knots.

This leads us to our first results:

\begin{quote}
\textbf{\textit{The existence of countably infinite many species of local conserved structures}} \\
There exist infinitely many species of local conserved structures. \\
\\
Any edge can be replaced by an edge with an isolated edge with some half-integer twist and a knot or a connected sum of knots. As there are infinitely many
half-integers and knots, there are therefore infinitely many such
species of structures.
\end{quote}

\begin{quote}
\textbf{\textit{Maximal number of local conserved quanities}} \\
For a closed 3-valent BRN with $N$ nodes, the maximum number of
ultra-local conserved quantities is $\frac{3N}{2}$.
\end{quote}

We can immediately lift these results to the scenario of un-framed
spin networks: excepting the twists, all the results follow
immediately. Additionally the first result does not depend on the
valence of the spin-network involved in any way and is therefore a
general result for embedded spin-networks.

\section{Conclusions and Discussion}

We have demonstrated the existence of a countable infinity of
species of local conserved structures within Braided Ribbon Networks
(and embedded spin-networks in general). We have also provided
several results of use for isolating these structures and
understanding when they are actually preserved.

In a theory of Quantum Gravity where the states are given by
spin-networks embedded in a 3-manifold all of these states will be
part of the Hilbert space.  The difficulty this poses comes from the fact that these locally conserved structures correspond to an infinite number of conserved quantities that don't correspond with anything that commutes with the constraints of general relativity.  This poses a significant problem in any attempt to recover the classical limit from a generic embedded spin-network - there is no reason to believe that these conserved quantities will simply cease to exist in the classical limit.  This leaves a significant dilemma: we must change something in the theory for general relativity to be the classical limit.

There are two immediately obvious alternatives for resolving this, the first being to modify the hamiltonian constraint in such a way that we introduce new generators or the evolution algebra.  The alternative to this is that we should reduce the physical Hilbert space of a theory of quantum gravity to require that there
do not exist any knots or linking. Our ability to consider this
super selection rule and still do certain things (including considering embedded spin networks) is questionable and
requires investigation before this can be adopted as an `easy'
solution.


\section*{Acknowledgements}

The authors are indebted to their Advisor, Lee Smolin, for his
discussion and critical comments. We thank Laurent Freidel, Sudance
Bilson-Thompson, Louis Kauffman and Isabeau Premont-Schwarz for helpful discussions.
We are further indebted to Sundance Bilson-Thompson for the use of
Fig. \ref{reducedlink} with permission.  Similarly Fig. \ref{trin1} is taken from \cite{Hackett2007} with permission. Research at Perimeter Institute for
Theoretical Physics is supported in part by the Government of Canada
through NSERC and by the Province of Ontario through MRI.


\begin{thebibliography}{99}
\bibitem{Penrose}
Roger Penrose, {\it Angular momentum: an approach to combinatorial space-time},
\textit{Quantum Theory and Beyond}, ed. T. Bastin, Cambridge University Press, Cambridge, 1971.


\bibitem{spin-foam}
  C.~Rovelli,
  {\it Living Rev.\ Rel.}\ {\bf 1}, 1 (1998)
  [arXiv:gr-qc/9710008].

\bibitem{spin-foam2}
  C.~Rovelli,
{\it Cambridge, UK: Univ. Pr. (2004) 455 p}

\bibitem{BilsonThompson:2005}
  S.~O.~Bilson-Thompson,
  arXiv:hep-ph/0503213.

\bibitem{BilsonThompson:2006}
  S.~O.~Bilson-Thompson, F.~Markopoulou and L.~Smolin,
  {\it Class.\ Quant.\ Grav.}\ {\bf 24}, 3975 (2007)
  [arXiv:hep-th/0603022].

\bibitem{Hackett2007}
  J.~Hackett,
  {\it Class.\ Quant.\ Grav.}\ {\bf 24}, 5757 (2007)
  [arXiv:hep-th/0702198].

\bibitem{Wan2007}
  Y.~Wan,
  arXiv:0710.1312 [hep-th].

\bibitem{BilsonThompson:2008ex}
  S.~Bilson-Thompson, J.~Hackett, L.~Kauffman and L.~Smolin,
  arXiv:0804.0037 [hep-th].

\bibitem{Wan:2008qs}
  Y.~Wan,
  arXiv:0809.4464 [hep-th].  Accepted for publication in \textit{Nucl. Phys. B}

\bibitem{LouNumber}
  S.~Bilson-Thompson, J.~Hackett and L.~H.~Kauffman,
  arXiv:0903.1376 [math.AT].


\bibitem{Hackett:2008tt}
  J.~Hackett and Y.~Wan,
  arXiv:0803.3203 [hep-th].

\bibitem{LeeWan2007}
  L.~Smolin and Y.~Wan,
  {\it Nucl.\ Phys.}\ B {\bf 796}, 331 (2008)
  [arXiv:0710.1548 [hep-th]].
\bibitem{He:2008is}
  S.~He and Y.~Wan,
  {\it Nucl.\ Phys.}\ B {\bf 804}, 286 (2008)
  [arXiv:0805.0453 [hep-th]].

\bibitem{He:2008jc}
  S.~He and Y.~Wan,
  {\it Nucl.\ Phys.}\ B {\bf 805}, 1 (2008)
  [arXiv:0805.1265 [hep-th]].




\bibitem{Major:1995yz}
  S.~Major and L.~Smolin,
  {\it Nucl.\ Phys.}\ B {\bf 473}, 267 (1996)
  [arXiv:gr-qc/9512020].

\bibitem{Borissov:1995cn}
  R.~Borissov, S.~Major and L.~Smolin,
  {\it Class.\ Quant.\ Grav.}\ {\bf 13}, 3183 (1996)
  [arXiv:gr-qc/9512043].

\bibitem{Smolin:2002sz}
  L.~Smolin,
  arXiv:hep-th/0209079.

\bibitem{Markopoulou:2008be}
  F.~Markopoulou and I.~Premont-Schwarz,
  {\it Class.\ Quant.\ Grav.}\ {\bf 25}, 205015 (2008)
  [arXiv:0805.3175 [gr-qc]].


\bibitem{Kribs:2005yc}
  D.~W.~Kribs and F.~Markopoulou,
  arXiv:gr-qc/0510052.

\bibitem{Smolin:1996fz}
  L.~Smolin,
  ``The classical limit and the form of the Hamiltonian constraint in
  non-perturbative quantum general relativity,''
  arXiv:gr-qc/9609034.











\end{thebibliography}
\end{document}